\documentclass[a4paper,UKenglish]{lipics}
\usepackage{amsthm}

\newtheorem{observation}{Observation}
\newtheorem{conjecture}{Conjecture}
\newcommand*\samethanks[1][\value{footnote}]{\footnotemark[#1]}

\begin{document}

\title{Orthogonal Vectors Indexing}
\date{}

\author[1]{Isaac Goldstein\thanks{This research is supported by the Adams Foundation of the Israel Academy of Sciences and Humanities}}
\author[1]{Moshe Lewenstein\thanks{This work was partially supported by an ISF grant \#1278/16}}
\author[1]{Ely Porat \samethanks}
\affil[1]{Bar-Ilan University, Ramat Gan, Israel \\  \texttt{\{goldshi,moshe,porately\}@cs.biu.ac.il}}

\authorrunning{I. Goldstein, M. Lewenstein and E. Porat}
\Copyright{Isaac Goldstein, Moshe Lewenstein and Ely Porat}

\maketitle

\thispagestyle{empty}

\begin{abstract}
In the recent years, intensive research work has been dedicated to prove conditional lower bounds in order to reveal the inner structure of the class P. These conditional lower bounds are based on many popular conjectures on well-studied problems. One of the most heavily used conjectures is the celebrated Strong Exponential Time Hypothesis (SETH). It turns out that conditional hardness proved based on SETH goes, in many cases, through an intermediate problem - the Orthogonal Vectors (OV) problem.

Almost all research work regarding conditional lower bound was concentrated on time complexity. Very little attention was directed toward space complexity. In a recent work, Goldstein et al.~\cite{GKLP17} set the stage for proving conditional lower bounds regarding space and its interplay with time. In this spirit, it is tempting to investigate the space complexity of a data structure variant of OV which is called \emph{OV indexing}. In this problem $n$ boolean vectors of size $c\log{n}$ are given for preprocessing. As a query, a vector $v$ is given and we are required to verify if there is an input vector that is orthogonal to it or not.

This OV indexing problem is interesting in its own, but it also likely to have strong implications on problems known to be conditionally hard, in terms of time complexity, based on OV. Having this in mind, we study OV indexing in this paper from many aspects. We give some space-efficient algorithms for the problem, show a tradeoff between space and query time, describe how to solve its reporting variant, shed light on an interesting connection between this problem and the well-studied SetDisjointness problem and demonstrate how it can be solved more efficiently on random input.
\end{abstract} 

\section{Introduction}
Recently, there is an intensive research work aimed at understanding the complexity within the class P (decision problems that are solved by polynomial time algorithms). Specifically, many conditional lower bounds have been proven on many polynomial algorithmic problems. These lower bounds are based on some conjectures on well-studied problems, especially notable are 3SUM, APSP and SETH.
The Strong Exponential Time Hypothesis (SETH)~\cite{IP01,IPZ01} states the following:

\begin{conjecture}
\textbf{Strong Exponential Time Hypothesis.} There is no $\epsilon>0$ such that $k$SAT can be solved in $O(2^{(1-\epsilon)n})$ for all $k$.
\end{conjecture}

Many conditional lower bounds for both polynomial and exponential time solvable problems are based on this conjecture. A partial list includes~\cite{PW10,LMS11,CPP13,RW13,AW14,AWW14,Bringmann14,FGLS14,BI15,ABW15,BK15,ABHWZ16,AWW16,MPS16}. For polynomial time solvable problems many of the conditional lower bounds are proven through the use of an intermediate problem called \emph{Orthogonal Vectors} (OV) which is defined as follows.

\begin{definition}
\textbf{Orthogonal Vectors.} Given a set $S$ of $n$ input vectors from $\{0,1\}^d$, decide if there are $u,v \in S$ such that $u$ is orthogonal to $v$.
\end{definition}

If SETH is true then there is no $O(n^{2-\epsilon})$ solution for OV for any $\epsilon>0$ (see~\cite{Williams05, WY14}). This conditional lower bound on OV was heavily used to obtain conditional lower bounds on the time complexity of a long list of algorithmic problems. This includes graph problems~\cite{RW13,MPS16}, dynamic problems~\cite{AW14}, string problems~\cite{AWW14,BI15,ABW15,BK15} and many other important problems from a variety of research fields.

A recent work by Goldstein et al.~\cite{GKLP17} set the stage for proving conditional lower bounds on space-time tradeoffs. Specifically, it was suggested that we can achieve space lower bounds by considering a data structure variant of SAT. Given a formula $\phi$ in a CNF format and a list of variables $L$ from $\phi$, we need to preprocess $\phi$ and $L$ and create a data structure to support the following queries. Given an assignment to all variables not in $L$ we are required to answer if this assignment can be completed to a full assignment that satisfies $\phi$. A closely related problem is \emph{Orthogonal Vectors Indexing} (OV Indexing) that is defined as follows.

\begin{definition}
\textbf{Orthogonal Vectors Indexing.} Given a set $S$ that contains $n$ $d$-length boolean vectors, preprocess $S$ and answer queries of the following form: given $d$-length boolean vector $v$, is there a vector in $S$ which is orthogonal to $v$.
\end{definition}

SETH can be reduced to OV indexing (see the details in Appendix~\ref{sec:hardness_of_OV_indexing}). As a consequence of this reduction there is no polynomial time preprocessing algorithm for OV indexing that achieves truly sublinear query time.

The main question that we consider is what the space requirements of OV indexing are. In this paper we examine this question in detail from various aspects for the case that $d=c\log{n}$ for some constant $c>1$ (if $c$ is non-constant its seems hard to achieve any improvement due to the connection to SETH). On one hand, solving OV indexing for input vectors of length $c\log{n}$ can be done easily using a lookup table of size $n^c$. Using this table, queries can be answered in constant time. On the other hand, without any preprocessing queries can be answered in linear time. It is interesting to figure out what can be done in between these two extremes. Can we achieve truly sublinear query time with less than $n^c$ space? Is there a clear tradeoff between time and space? What can we say about the reporting version of this problem? In this paper we investigate all these questions and more.

Understanding the space requirements of OV indexing is interesting in its own right, but it can have many implications on other problems. Along the lines of Goldstein et al.~\cite{GKLP17} OV indexing can serve as a basis for proving conditional hardness in terms of space for other algorithmic problems. Specifically, as OV is a standard tool in demonstrating conditional hardness of problems in terms of time it is likely that understanding the space hardness of its data structure variant - OV indexing - can be applied to many problems shown to be hard based on OV. In the work by Goldstein et al.~\cite{GKLP17} there was an attempt to state a general hardness conjecture for OV indexing. However, as no solution to neither OV indexing nor the data structure variant of SAT was suggested in~\cite{GKLP17} (other than the trivial ones), a more fine grained conjecture was out of reach. One major motivation for this paper is to state such a conjecture based on improved upper bounds for OV indexing (see more detailed discussion in the last section of this paper).

\medskip

\textbf{Related Work.} The Partial Match problem and its variants were extensively studied for decades. These problems are related to our OV indexing problem (see, for example,~\cite{AWY15}). One of the first works regarding Partial Match is by Rivest~\cite{Rivest74,Rivest76}. However, his work focused on the average case analysis of several solutions for the problem that in the worst case do not achieve an improvement over the trivial solution, unless the number of "don't cares" symbols (corresponding to the zeroes in the OV indexing problem) in the query is not too large. Many works on the Partial Match problem and its variants focus on improving the time complexity rather than the space complexity which is the main concern of this paper. Other works that do consider space complexity deal with the case of very large dimension $d$ that can be even linear in $n$~\cite{CIP02,CGL04,LW17}. This case admits very different behaviour from the case we handle in this paper in which $d=\Theta(\log{n})$.

\medskip

\textbf{Our Results.} In this paper we present the following results regarding OV indexing. We suggest 3 algorithms that solve OV indexing with truly less that $n^c$ space and truly sublinear query time. We show how to use the second and third algorithms we present to get a tradeoff between space and query time. A variant of the first algorithm is used to prove the connection between OV indexing and SetDisjointness, a problem which was considered by several papers as the basis for showing space conditional hardness.
We also solve the reporting variant of OV indexing in which we need to report all input vectors that are orthogonal to our query vector. Finally, we show that, on random input vectors, OV indexing can be solved more efficiently in terms of space.

\section{DivideByOnes: First Space-Efficient Solution for OV indexing}\label{sec:first_solution}
Our goal is to achieve an algorithm that has truly sublinear query time and requires $O(n^{c-\epsilon})$ space for some $\epsilon>0$. This is an improvement over the trivial algorithm that uses $n^c$ space. We note that in this solution and throughout this paper the notations $\tilde{O}$ and $\tilde{\Omega}$ (almost always) suppress not just polylogarithmic factors as usual, but also all factors that are smaller than $n^{\epsilon}$ for any $\epsilon>0$.

\subsection{DivideByOnes Algorithm}

\textbf{Preprocessing}. The first step is to save a set $S_1$ of all vectors from $S$ with at most $c_1\log{n}$ ones for some constant $0<c_1\leq c/2$. There are at most $\sum_{k=0}^{c_1\log{n}}\binom{c\log{n}}{k} \leq c_1\log{n}\binom{c\log{n}}{c_1\log{n}}$ vectors in $S_1$. We have that $\binom{c\log{n}}{c_1\log{n}} \approx n^{c\log{c}-c_1\log{c_1}-(c-c_1)\log{(c-c_1)}}$ (see Appendix~\ref{sec:approx_binomial}). We choose the largest $c_1$ such that the number of vectors in $S_1$ will be $\tilde{O}(n^{1-\epsilon})$ for some $\epsilon>0$.

Let $S_2$ be the set of vectors from $S$ with more than $c_1\log{n}$ ones.
Assume that $c$ is an integer. We split each vector in $S_2$ into $c$ parts each of length $\log{n}$ bits. As all the vectors in $S_2$ have at least $c_1\log{n}$ ones, we are guaranteed that at least one of the $c$ parts of each vector has at least $\frac{c_1}{c}\log{n}$ ones.

We have $n$ possible vectors of size $\log{n}$, so we create $c$ arrays $A_1,A_2,...,A_c$ of length $n$ each, such that the $i$th entry in each array represents the $\log{n}$-length boolean vector that its numerical value is $i$. In the $i$th entry of an array $A_j$ we create a list that contains each vector $v \in S_2$ such that: (i) The number of ones it has in its $j$th part is the maximum among all its parts (ties are broken arbitrarily). (ii) The value of its bits in its $j$th part is orthogonal to the value of the $\log{n}$-length vector whose numerical value is $i$ (the numerical value of an $m$-length vector is the value of this boolean vector that is parsed as an $m$-length boolean number).

To analyse the space consumed by these arrays one should notice that each vector $v \in S_2$ appears only in one array. Moreover, as $v$ appears only in the array that represents the part in which $v$ has the maximum number of ones, the number of lists in this array that contain $v$ is at most $\frac{n}{2^{\frac{c_1}{c}\log{n}}}=n^{1- \frac{c_1}{c}}$. Therefore, the total size of all arrays is no more than $n \cdot n^{1- \frac{c_1}{c}} = n^{2- \frac{c_1}{c}}$ which is truly subquadratic.

\textbf{Query}. When we get a query vector $u$ we first check in $S_1$ if there is a vector that is orthogonal to $u$. Then, we partition $u$ to $c$ equal parts. For each part $j$ if the numerical value of all bits in this part is $i$ we check all the vectors in the $i$th list of $A_j$ and verify if one of them is indeed orthogonal to $u$. The problem with this process is that the length of the list we check may be $\tilde{\Omega}(n)$, so our query time will be $O(n)$ which is trivial. To overcome this and obtain a constant query time for long lists, we need to treat lists whose length is $\tilde{\Omega}(n)$ differently in the preprocessing phase.

\textbf{Additional Preprocessing}. For each entry $i$ in some array that the length of the vectors list in it is not truly sublinear, we store a bitmap that tells for all possible values of the other $(c-1)\log{n}$ bits whether there is a vector in $S$ that is orthogonal to these bits and the $\log{n}$ bits represented by $i$. The size of the bitmap is $2^{(c-1)\log{n}}=n^{c-1}$. As calculated before, the total number of vectors in all lists of the array is $n^{2- \frac{c_1}{c}}$. Consequently, the number of lists that have $\tilde{\Omega}(n)$ vectors in them is no more than $n^{1- \frac{c_1}{c}}$. Therefore, the space needed for all bitmaps is $n^{c - \frac{c_1}{c}}$.

\subsubsection{Generalization to klogn}
We can generalize the above solution by partitioning the vectors to parts whose size is $k\log{n}$ for some $k>0$. First we consider the case that $k$ divides $c$. In this case, the algorithm continues in same way as for the case that $k=1$. The number of lists in each array is $n^k$. Each input vector $v$ has at least $\frac{c_1}{c}k$ ones in the part with the largest number of ones. Consequently, each input vector $v$ occurs in $n^{k-\frac{c_1}{c}k}$ lists in the array corresponding to the part with most ones in $v$. The total size of all arrays and lists is $O(n^{k+1- \frac{c_1}{c}k})$. The number of long lists is at most $O(n^{k-\frac{c_1}{c}k})$. Each bitmap has size $n^{c-k}$. Therefore, the space usage for handling long lists is $O(n^{c-\frac{c_1}{c}k})$. The total space of the data structure is $O(n^{k+1-\frac{c_1}{c}k}+n^{c-\frac{c_1}{c}k})$. By setting $k=c-1$ (if possible, otherwise see the next paragraph) we get the lowest space complexity, which is $O(n^{c-c_1(1-\frac{1}{c})})$.

In case $k$ does not divide $c$, we can partition each vector to $\lfloor \frac{c}{k} \rfloor$ parts of length $k\log{n}$. However, we are left with one part $P$ whose length is smaller than $k\log{n}$. It can be the case that for some input vector $v$ the number of ones in each of the parts of length $k\log{n}$ is smaller than $\frac{c_1}{c}k$, as there can be many ones in $P$. In order to solve this problem we can do the following. Let $c_1' = c_1-\epsilon$ for any $\epsilon>0$. We define $k'=\lfloor \frac{k}{\epsilon} \rfloor \epsilon$ and $c'=\lfloor \frac{c}{\epsilon} \rfloor \epsilon$. It is clear that $k' > k-\epsilon$ and $c' > c -\epsilon$. Each input vector $v$ can be partitioned to $m_1=\lfloor \frac{c}{\epsilon} \rfloor$ parts $P_1,P_2,...,P_{m_1}$ whose length is $\epsilon$ and another optional part $P$ whose length is less than $\epsilon$. If we ignore the bits of any vector in $P$, we are still guaranteed that there are at least $c_1'$ ones in the rest of the vector. We can choose $m_2=\lfloor \frac{k}{\epsilon} \rfloor$ parts from the $m_1$ parts $P_1,P_2,...,P_{m_1}$. This will give us exactly $k'\log{n}$ bits. There are $m_3=\binom{m_1}{m_2}$ options of how to choose $m_2$ parts out of the $m_1$ parts. The number $m_3$ is constant as $k$, $c$ and $\epsilon$ are all constants. Therefore, we can create $m_3$ arrays $A_1,A_2,...,A_{m_3}$ each one of them represents some $k'\log{n}$ bits from our input vectors. We handle these arrays as in the regular case explained above. The crucial point one should observe is that for each input vector $v$ there must be $k'\log{n}$ bits among these $m_3$ options that contains at least $\frac{c_1'}{c}k$ of the ones in $v$. Let $A_i$ be the array representing $k'\log{n}$ bits out of the $m_3$ options that contains the maximum number of ones in $v$. We are guaranteed that $v$ will appear in at most $n^{k'-\frac{c_1'}{c}k'}$ lists in $A_i$. We continue the solution as in the regular case. Following the analysis of the regular case, we have that the total space of the data structure will be $O(n^{k'+1-\frac{c_1'}{c}k'}+n^{c-\frac{c_1'}{c}k'})$. As $k-\epsilon < k' \leq k$ and $c_1'=c_1-\epsilon$, we get that the total space is $O(n^{k+1-\frac{c_1-\epsilon}{c}(k-\epsilon)}+n^{c-\frac{c_1-\epsilon}{c}(k-\epsilon)})$. Setting $k=c-1$ as before, we get that the space is $O(n^{c-\frac{c_1-\epsilon}{c}(c-1-\epsilon)})$. We can make this space complexity as close as we wish to the space complexity for the case $k$ divides $c$ by choosing $\epsilon$ whose value is very close to 0. Consequently, we have the following result ($c_1$ is the largest number that satisfies $\binom{c\log{n}}{c_1\log{n}}=\tilde{O}(n^{1-\delta})$ for some $\delta>0$):

\begin{theorem}
For every $\epsilon > 0$ the DivideByOnes algorithm solves OV indexing with truly sublinear query time using $O(n^{c-\frac{c_1-\epsilon}{c}(c-1-\epsilon)})$ space.
\end{theorem} 

\section{TopLevelsQueryGraph: Second Space-Efficient Solution for OV indexing}\label{sec:second_solution}
There are two problems with the previous solution. The first one is the sharp separation between long lists (having $\tilde{O}(n)$ vectors) and short lists. For long lists we use a large amount of space and answer queries very quickly in constant time, while for short lists we just save the vectors in the lists and spend time in the query stage. The second problem is that each input vector is saved many times in different lists.

\subsection{Query Graph}
In order to improve the space requirements for sublinear query time we introduce the notion of a \emph{query graph}. The idea of the query graph is to create a tradeoff between query time and space and save each vector just once. We are now ready to define the query graph. A query graph is a directed acyclic graph $G=(V,E)$ such that each vertex $v_i$ in $V$ represents a boolean vector $\alpha_i$ of length $k\log{n}$. There is an edge $(v_i,v_j) \in E$ if the vectors $\alpha_i$ and $\alpha_j$ differ on exactly one element which is 0 in $\alpha_i$ and 1 in $\alpha_j$. Following this definition the query graph can be viewed as a layered graph with $k\log{n}$ layers. The $j$th layer in this graph contains all the nodes $v_i$ such that the number of ones in $\alpha_i$ is exactly $j$. All the edges from the vertices in the $j$th layer are directed to vertices in the $(j+1)$th layer. We call the layers for small values of $j$ \emph{top} layers and the layers with high values of $j$ \emph{bottom} layers.

Let $W$ be a set of indices such that $W \subseteq [c\log{n}]$ and $|W| = k\log{n}$. We want each vertex $v_i$ that represents a vector $\alpha_i$ to contain a list $L_i$ of input vectors such that their elements in the indices specified by $W$ are orthogonal to $\alpha_i$. This is the same as we did in the previous construction as each entry in an array contains all input vectors that are orthogonal to the value of this entry in indices of the relevant part. However, instead of saving all input vectors that are orthogonal to $\alpha_i$ in the indices specified by $W$, we just pick all the input vectors that their elements in the indices specified by $W$ are exactly the complements of the elements in $\alpha_i$. All these vectors are saved in the list $L_i$ in vertex $v_i$. Using these lists, we have the following easy observation:

\begin{observation}
Given a set $W \subseteq [c\log{n}]$ such that $|W| = k\log{n}$, the complete list of input vectors such that their values in the indices specified by $W$ are orthogonal to some $\alpha_i$ can be recovered by concatenating all lists of vectors in the vertices that are reachable from vertex $v_i$ in the query graph $G$.
\end{observation}

\subsection{TopLevelsQueryGraph Algorithm}
We start the preprocessing phase by constructing a query graph $G$. Now, following the last observation, instead of saving each vector many times in all the lists that their index is orthogonal to our query in the relevant indices (as suggested by the previous solution), we can save each vector in just one list and recover the original list by traversing $G$. We start the traversal from the vertex $v_i$ such that the values of the query vector in the indices specified by $W$ are equal to $\alpha_i$. We can use any standard graph traversal algorithm to obtain all the input vectors that are orthogonal to the query vector in the indices specified by $W$. The number of vertices that we visit in the traversal of the query graph for a query vector $q$ that have $k'\log{n}$ ones in the indices specified by $W$ is $2^{k\log{n}-k'\log{n}} = n^{k-k'}$.

We can identify two types of nodes in the query graph. A node $v_i$ that has an empty list $L_i$ is considered a \emph{black node}, otherwise it is considered a \emph{white node}. We note that the number of white nodes is at most $n$ and it can be $O(n)$ if the input vectors are split between many lists.
In order to achieve a truly sublinear query time we would like the number of nodes we visit during the traversal in the query graph to be truly sublinear. Moreover, as the number of white nodes can be $\Theta(n)$ we need to make sure that the total number of white nodes we visit is truly sublinear even if we know how to avoid black nodes. As mentioned before, the number of nodes we visit during our traversal is $n^{k-k'}$ which is truly sublinear if we set $k-k'<1$. This means that we need to handle queries that match some vertex $v_i$ in the top levels of the query graph differently. For all vertices $v_i$ in the $x$ top levels of the graph we create a list $L'_i$ of \emph{all} input vectors that are orthogonal to $\alpha_i$. Then, for each list $L'_i$ we create a bitmap to quickly identify if there is a vector in the list that is orthogonal to our query. The size of each bitmap is $n^{c-k}$. The total number of bitmaps we create is $\tilde{O}(\binom{k\log{n}}{x\log{n}})$ for $x \leq k/2$ as the number of vectors in the $j$th level of the query graph is $\binom{k\log{n}}{j\log{n}}$ (we choose $j\log{n}$ positions for the ones in $\alpha_i$ out of $k\log{n}$ positions). Moreover, the number of layers is logarithmic in $n$. Thus, the total required space for handling the top layers of the query graph is $\tilde{O}(n^{c-k}\binom{k\log{n}}{x\log{n}})$. The binomial coefficient $\binom{k\log{n}}{x\log{n}}$ can be approximated by $n^{k\log{k}-x\log{x}-(k-x)\log{(k-x)}}$ using Stirling's approximation (see Appendix~\ref{sec:approx_binomial}). So, the total space for the top layers is approximately $\tilde{O}(n^{c-k+k\log{k}-x\log{x}-(k-x)\log{(k-x)}})$.

Now, a query vector $q$ that matches a vertex $v_i$ in the $x$ top levels can be answered in constant time by just looking at the proper entry in the bitmap of $v_i$. Otherwise, the number of vertices we need to traverse in the query graph will be at most $n^{k-x}$ which is truly sublinear if $k-x<1$. The problem is that the total number of vectors in the lists of these vertices can be $\theta(n)$. To overcome this problem, we change the way we handle any list $L_i$ in the $(k-x)\log{n}$ bottom levels according to the number of elements in it. If the number of elements in the list is $O(n^{1-k+x})$ we do nothing - the elements are kept in the list with no special treatment. Otherwise, we save a bitmap over all the possibilities of the other bits in the query vector. The size of the bitmap, as before, is $n^{c-k}$. The number of lists that have more than $O(n^{1-k+x})$ elements is at most $n^{k-x}$. Therefore, the space for all the bitmaps of the long lists is $n^{c-x}$. We have that the total space of our data structure is $\tilde{O}(n^{c-k+k\log{k}-x\log{x}-(k-x)\log{(k-x)}}+n^{c-x})$. To obtain the best space complexity (while preserving the truly sublinear query time), we set $k$ very close to $1.3$ and $x$ to $0.3$. The space complexity of this solution using these values is approximately $\tilde{O}(n^{c - 0.3})$. To conclude, we obtain the following result:

\begin{theorem}
The TopLevelsQueryGraph algorithm solves OV indexing with truly sublinear query time using approximately $\tilde{O}(n^{c - 0.3})$ space.
\end{theorem}

\section{BottomLevelsQueryGraph: Third Space-Efficient Solution for OV indexing}
We can use the query graph to obtain another solution to the OV indexing problem. This time we focus on the $x$ bottom levels of the query graph. For each vertex $v_i$ in the $x$ bottom levels of the query graph we save a bitmap to quickly identify if there is an input vector such that (a) Its bits in the indices specified by $W$ are the complements of $\alpha_i$ and (b) It is orthogonal to our query vector. The space we invest in these bitmaps is $\tilde{O}(n^{c-k}\binom{k\log{n}}{x\log{n}})$. Then, for every vertex $v_i$ which is not in the $x$ bottom levels of the query graph we save in its list $L_i$ all the input vectors that are \emph{orthogonal} to $\alpha_i$, but \emph{do not appear} in the any of the lists of the vertices in the $x$ bottom. For every list $L_i$ that its length is $\tilde{\theta}(n)$ we save a bitmap to get the answer in $\tilde{O}(1)$ time. Because we do not include in any list $L_i$ vectors from the lists in the $x$ bottom levels, we are guaranteed that each input vector appears in at most $n^{k-x}$ lists. In our view of the query graph, this means that if an input vector appears in the list $L_i$ of some vertex $v_i$ it will be duplicated in the lists of all vertices that $v_i$ is reachable from them. Consequently, the total number of vectors in all lists above the $x$ bottom levels is at most $n^{1+k-x}$. Therefore, the number of bitmaps we will save for lists of size $\tilde{\theta}(n)$ is at most $n^{k-x}$. Each bitmap is of $n^{c-k}$ space, so the size of all bitmaps is $n^{c-x}$. The total size of the data structure is again $\tilde{O}(n^{c-k+k\log{k}-x\log{x}-(k-x)\log{(k-x)}}+n^{c-x})$.

Upon receiving a query vector $q$, if it matches a vertex in one of the $x$ bottom levels, we immediately get the answer by looking at the right entry in the bitmap in that vertex. Otherwise, we need to look not just at the bitmap of the vertex that matches our query, but rather we have to go over all the vertices $v_i$ in the $(k-x)\log{n}$ level (the top level of the $x$ bottom levels) such that $\alpha_i$ is orthogonal to $q$ in the positions specified by $W$. In all these vertices we check in their bitmap if there is an input vector that is orthogonal to $q$. If $k-x<1$ we ensure that the query time is sublinear in $n$. All in all, we obtain a solution that has the same query time and space complexities as the previous one using a different approach, as summarized in the following theorem:

\begin{theorem}
The BottomLevelsQueryGraph algorithm solves OV indexing with truly sublinear query time using approximately $\tilde{O}(n^{c - 0.3})$ space.
\end{theorem}

\section{Space and Query Time Tradeoff for Solving OV indexing}
In all the solutions we presented so far we tried to minimize the space usage and still achieve a sublinear query time. However, obtaining a tradeoff between the space and query time would be of utmost interest. We know how to obtain constant query time by using $n^c$ space. But can we obtain, for example, $O(\sqrt{n})$ query time using just $n^{c-\epsilon}$ space for some $\epsilon>0$? In the first method we have suggested there is an inherent problem to achieve this as all lists can have more than $O(\sqrt{n})$ vectors. In the second and third solutions we can improve the query time by choosing larger $x$. However, as $x$ becomes $k/2$ the space of the data structure becomes $\tilde{O}(n^c)$. The following theorem demonstrates how to obtain any polynomial query time while consuming $O(n^{c-\gamma})$ space for some $\gamma>0$.

\begin{theorem}
For any $\epsilon>0$ there is a solution to OV indexing that its query time is $O(n^{\epsilon})$ and the space complexity is $O(n^{c-\gamma})$ for $\gamma>0$.
\end{theorem}

\begin{proof}
The idea is to combine the second and third solutions. We can save bitmaps for both the $x$ top levels and the $x$ bottom levels of the query graph using $\tilde{O}(n^{c-k}\binom{k\log{n}}{x\log{n}})$ space. Then, for every vertex that is not in the $x$ top or bottom levels we do the same as in the second solution - save a bitmap for every node whose list is of length $\theta(n^{\delta})$ or more for some $\delta>0$. The total cost of these bitmaps is $\tilde{O}(n^{c-k+1-\delta})$. When we get a query vector $q$ that matches a vertex $v_i$ in our query graph. If $v_i$ is on the $x$ top or bottom levels, we just check the right entry in the bitmap of $v_i$. Otherwise, we start a traversal from $v_i$ to all the vertices that are reachable from it except those in the $x$ bottom levels. The number of vertices we visit is at most $\binom{(k-x)\log{n}}{x\log{n}}$ if $k/3 < x$. This is approximately $\tilde{O}(n^{(k-x)\log{k-x}-x\log{x}-(k-2x)\log{(k-2x)}})$. It is easy to verify that as $x$ gets close to $k/2$ the exponent of this expression is very close to 0. Therefore, the total query time is  $\tilde{O}(n^{\delta+(k-x)\log{(k-x)}-x\log{x}-(k-2x)\log{(k-2x)}})$ as the query time in each vertex we visit is at most $n^{\delta}$. By choosing suitable value of $k \geq 1$, $x<k/2$ and $\delta>0$, we can obtain a query time of $\tilde{O}(n^{\epsilon})$ for any $\epsilon>0$ using a data structure that consumes $\tilde{O}(n^{c-\gamma})$ space for some constant $\gamma>0$.
\end{proof}

\section{The Reporting Version of OV indexing}
In the reporting version of OV indexing, given a query vector $q$ we are required not just to decide if there is a vector in $S$ that is orthogonal to $q$, but rather we are required to report all input vectors in $S$ that are orthogonal to $q$.

To solve this version we can use the same methods as we have described for the decision version. However, the only part of these solutions that does not support reporting is the use of bitmaps. Using a bitmap we can answer the query quickly if there is an input vector that is orthogonal to our query vector, but we are unable to discover the list of input vectors that are orthogonal to the query if there are such vectors.
The following lemma demonstrates how to construct a data structure that uses almost the same space as a bitmap, but supports efficient reporting.

\begin{lemma}
Given $n$ $c\log{n}$-length boolean vectors, there is a data structure that uses $\tilde{O}(n^c)$ preprocessing time and upon receiving a query vector $v$ report on all $t$ input vectors that are orthogonal to $v$ in time $O(t\log{n})$
\end{lemma}

\begin{proof}
In order to report these vectors we build the following data structure. The data structure is composed of a complete binary tree with $c\log{n}$ levels (we consider the root of the tree as level number $0$). Every edge in the tree has a label. The label of an edge to the left child is 0 and the label to the right child is 1. Moreover, in each tree node in the $i$th level of the tree we save a bitmap of size $n^{c}/2^i$. The $j$th bit of a bitmap in the $i$th level of the tree represents the $j$th boolean vector in the lexicographical order of all $(c\log{n}-i)$-length boolean vectors. We denote this vector by $v_{i,j}$. The $j$th bit of the a bitmap in a node $x$ in the $i$th level of the tree will be 1 if there is an input vector that its prefix has the same values as in the string obtained by concatenating all the labels of the edges on the path from the root to $x$ and that its last $(c\log{n}-i)$ elements are orthogonal to $v_{i,j}$. Otherwise, the value of this bit will be 0. The space required by this data structure is clearly $\tilde{O}(n^c)$, as we have $c\log{n}$ levels, and there are $2^i$ bitmaps in the $i$th level each of them of size $n^c/2^i$.

Upon receiving a query vector $q$, the bit in the bitmap of the root that represents this vector is observed. If it is 1, then we know that there is at least one input vector that is orthogonal to $q$. We examine, for both children of the root, the bit that represents the last $c\log{n}-1$ bits of $q$. If the value of this bit is 1, we recurse on this child node. Otherwise, we stop searching in this path. When we get to a leaf $x$ and the value of the single bit in this leaf is 1, the vector that is given by concatenating the bits on the labels of the edges along the path from the root to $x$ is an input vector that is orthogonal to $q$. Equivalently, we can save a pointer from every leaf having 1 in the single bit it contains to the input vector it represents and use this pointer for reporting. The total query time is $O(t\log{n})$, where $t$ is the number of input vectors that are orthogonal to our query vector.
\end{proof}

\subsection{Improving The Query Time}
We can remove the dependency on $n$ in the query time as shown by the following theorem.

\begin{theorem}
Given $n$ $c\log{n}$-length boolean vectors, there is a data structure that uses $\tilde{O}(n^c)$ space and upon receiving a query vector $v$ report on all $t$ input vectors that are orthogonal to $v$ in $O(t)$ time.
\end{theorem}

\begin{proof}
From every bit that has value 1 and represents a vector $v$ in the bitmap of the root node, we create a pointer to a bit in a bitmap of a descendant node if it satisfies the following conditions: (1) The bit value is 1 and it represents a vector $u$ that is a suffix of $v$. (2) Both children of this descendant node have 1 in the bit that represent a vector $w$ that is a suffix of $u$. (3) There is no bit in a bitmap of some other descendant node that satisfies both previous conditions and represents a vector $u'$ such that $u$ is a suffix of $u'$. If the bit is in a leaf node then we need to satisfy only the first and third conditions. After we finish adding pointers from bits in the bitmap of the root node, we recurse on every bit that we point to unless it is a bit in a leaf. At the end of this process, we use the pointers that we have created to construct the following data structure.

We keep only the bitmap of the root node. We call this bitmap $B$. For every bit whose value is 1 in $B$ we create a node. Moreover, we create a node for every bit that is pointed to by a pointer that we have created previously. Therefore, each node we create \emph{represents} a specific bit in some bitmap. We also create an edge $(u,v)$ between two nodes $u$ and $v$ if there is a pointer from the bit that node $u$ represents to the bit that node $v$ represents. We keep pointers from the ones in $B$ to the nodes that represents them. We also have pointers from leaf nodes to the input vectors they represent.

Using this data structure queries can be answered more quickly. Specifically, upon receiving a query vector $q$, we check the value of the bit that corresponds to this vector in the bitmap of the data structure. If the value is 0 we are done - no input vector is orthogonal to $q$. Otherwise, we follow the pointer from this bit to the node that represents it and use any tree traversal algorithm to get to all the leaves reachable from this node. Then, the pointers in these leaves leads us to all the input vectors that are orthogonal to $q$. It is easy to verify that the query time using this data structure is just $O(t)$.

Now, we show that the size of the data structure is still $\tilde{O}(n^c)$. The number of nodes we created is bounded by the number of the ones in all the bitmaps of the previous structure. Moreover, every edge $(u,v)$ in the data structure corresponds to a path of bits whose value is 1 that start at the bitmap that contains the bit that is represented by $u$ and ends at the bitmap that contains the bit that is represented by $v$. If we look at some specific node $v$ and all the edges that point to $v$ and their corresponding paths, they create a binary tree whose inner nodes, except the root, represent bits that their value is 1 and do not correspond to any node in our data structure. Moreover, these bits do not belong to any other binary tree. Otherwise, they must have two children 1 bits and become nodes in our structure. Therefore the total number of edges in our structure is $O(\ell)$, where $\ell$ is the number of bits with value 1 in all bitmaps. Consequently, the size of the new data structure is $\tilde{O}(n^c)$ - the same bound as for the previous structure, but the query time is just $O(t)$ which is optimal.
\end{proof}

We can plug in the data structure from the previous theorem into any of the three solutions for OV indexing and get solutions for the reporting version of OV indexing that have the same space usage (up to logarithmic factors) and just an additive $O(t)$ to the query time.

\section{Reducing OV indexing to SetDisjointness}
In this section we present a connection between OV indexing and the problem of SetDisjointness. In the problem of SetDisjointness we are given $m$ sets $S_1,S_2,...,S_m$ such that the total number of elements in all sets is $N$ and after preprocessing them we need to answer queries of the following form: given a pair of indices $(i,j)$, decide whether $S_i \cap S_j$ is empty or not. The problem can be generalized to $k$-SetDisjointness in which we are given as a query a $k$-tuple $(i_1,i_2,...,i_k)$ and we are required to answer if the intersection $S_{i_1} \cap S_{i_2} \cap ... \cap S_{i_k}$ is empty or not. The SetDisjointness problem was the first problem used to show conditional lower bounds on space complexity(see~\cite{CP10,PRT12,DSW12,PR14}). Therefore, it should be interesting to see the connection between our OV indexing problem and the fundamental problem of SetDisjointness. Other problems connected to SetDisjointness are discussed in~\cite{GKLP17}. Currently, the best known space-query time tradeoff for $k$-SetDisjointness is $S \times T^k = O(N^k)$, where $S$ is the space complexity and $T$ is the query time~\cite{CP10a,GKLP17}.

We begin by presenting a simple reduction from OV indexing to $k$-SetDisjointness for $k=c\log{n}$. Given an instance of OV indexing with $n$ $c\log{n}$-length boolean input vectors we can create an instance of $k$-SetDisjointness in the following way. We create $c\log{n}$ sets. The set $S_i$ contains all the vectors that have 0 in their $i$th element. Then, given a query vector $q$ that has ones in the elements whose indices are $(i_1,i_2,...,i_k)$ all that we need in order to answer this query is to verify if the intersection $S_{i_1} \cap S_{i_2} \cap ... \cap S_{i_k}$ is empty or not. If it is empty then we know that there is no input vector that has zeroes in all the position of the ones in $q$, which means that no input vector is orthogonal to $q$. Otherwise, there is an input vector which is orthogonal to $q$.

We would like to show this reduction to other values of $k$, especially small and constant. The idea is to use the first solution that we have suggested to obtain the following result:

\begin{theorem}
There is a reduction from OV indexing to $k$-SetDisjointness that can be used to solve OV indexing with truly sublinear query time and $O(n^{c-\gamma})$ space for some $0<\gamma<1$.
\end{theorem}

\begin{proof}
In the first solution called DivideByOnes, we have a set $S_1$ of all input vectors with a small amount of ones (less than $c_1 \log{n}$) and arrays $A_1,A_2,...,A_c$ such that each of them corresponds to $\log{n}$ bits of each input vector (to simplify the discussion we consider the basic case of Section~\ref{sec:first_solution} in which every part is of length $\log{n}$ and $c$ is an integer. However, the results here can be adapted to handle the general case). We are guaranteed that for every input vector $v$ there is at least one part of $\log{n}$ bits that the number of ones in it is at least $\frac{c_1}{c}\log{n}$. The array that corresponds to this part is the only one in which we place $v$. Now, instead of placing vector $v$ only in this specific array we would like to place it and every other input vector in all arrays. Then, for each entry in every array we create a set containing all the vectors in the list at that entry. This sets form an instance of $k$-SetDisjointness where $k=c$. Given a query vector $q$ we can answer the query as follows: Split $q$ into $c$ parts of length $\log{n}$. Find the entry that corresponds to the value of $q$ in each of these parts. Intersect all the sets that correspond to these entries. If the intersection is empty there is no orthogonal input vector to our query vector. Otherwise, there is at least one input vector that is orthogonal to $q$.

The number of elements in all sets can be $n^2$, as each vector can occur in $O(n)$ entries in all arrays. Therefore, even for $k=2$ we will need at least $n^2$ space (linear in the number of elements) and the query time will be linear in $n$. This query time is trivial, and we can improve it by consuming more than $n^2$ space. However, this is also trivial as we can save a lookup table instead.

To improve this tradeoff we reduce the number of elements in the sets. In the solution that we have described we save the set $S_1$ of input vectors with at most $c_1 \log{n}$ ones. However, when we partition the vectors there is no guarantee on the number of ones in each part. Consequently, input vectors can appear in many of the sets for some part. To prevent that we wish to find a partition of the input vectors such that the number of ones in each part will be more or less the same. It turns out that if we take a random partition then the number of sets that contains a specific vector $v$ is expected to be truly sublinear. Let $X_j$ be a random variable representing the number of sets that contain an input vector $v_j$. We count only sets that correspond to lists in an array for the same part of the vector. The $\log{n}$ positions of the bits of that part are chosen randomly (we note that there is no difference in the analysis of different parts as the partition is random). Let $x\log{n}$ be the number of ones in $v_j$. There are $\binom{c\log{n}}{\log{n}}$ options to choose $\log{n}$ bits out of the total $c\log{n}$ bits of $v_j$. The number of options to choose $\log{n}$ bits such that the number of zeros is exactly $i$ is $\binom{(c-x)\log{n}}{i}\binom{x\log{n}}{\log{n}-i}$. Therefore, we have $\Pr(X_j=2^i) = \frac{\binom{(c-x)\log{n}}{i}\binom{x\log{n}}{\log{n}-i}}{\binom{c\log{n}}{\log{n}}}$, as $v_j$ appears in $2^i$ sets if it has $i$ zeroes in the part they represent. Consequently, the expected value of $X_j$ is $\mathbf{E}(X_j) = \Sigma_{i=\log{n}-x\log{n}}^{\log{n}}\frac{\binom{(c-x)\log{n}}{i}\binom{x\log{n}}{\log{n}-i}2^i}{\binom{c\log{n}}{\log{n}}}$. If $i=\log{n}-y$ such that $y=\Theta(\log{n})$ then the value of $2^i$ is subpolynomial. Otherwise, if $y=o(\log{n})$, then the value of $\binom{x\log{n}}{\log{n}-i} = \binom{x\log{n}}{y}$ is subpolynomial while $\frac{\binom{(c-x)\log{n}}{\log{n}-y}}{\binom{c\log{n}}{\log{n}}}=O(1/n^\epsilon)$ for some $\epsilon>0$. Therefore, each of the elements in the sum of $\mathbf{E}(X_j)$ is truly sublinear. Moreover, there are at most $\log{n}$ elements in this sum. Consequently, $\mathbf{E}(X_j)$ is truly sublinear. All input vectors not in $S_1$ have at least $c_1\log{n}$ ones, so as we have proven the expected number of sets that contain each one of them is truly sublinear. By linearity of expectation we get that the expected total number of elements in all sets for some part will be truly subquadratic. As the number of parts is constant, the total number of elements in all sets is truly subquadratic.

Following the above discussion, there is a reduction from OV indexing to $k$-SetDijointness such that the total number of elements is subquadratic. Similar analysis can be done for the general case in which we partition each vector to parts that their size is not $\log{n}$. Therefore, if we partition each vector to $k$ equal parts (in case $k$ does not divide $c$ there might be parts that have one more bit than the others) then the number of lists in each array is $n^{c/k}$. Consequently, if every input vectors occurs in all lists the total number of elements in all sets created by our reduction is $N=n^{c/k+1}$. We can plug this value of $N$ in the tradeoff between the query time and space for $k$-SetDisjointness which is $S \times T^k = O(N^k)$ and conclude that to obtain truly sublinear query time for OV indexing following our reduction no less than $n^c$ space is needed. However, by using a random partition and following a similar analysis to that of the specific case considered before (where each part is of size $\log{n}$), the total number of elements in all sets is expected to be just $N=n^{c/k+1-\delta}$ for some $\delta>0$. Now, if we plug this value of $N$ in the tradeoff between the query time and space for $k$-SetDisjointness we obtain our result.
\end{proof}

\section{OV indexing for Random Input}\label{sec:random_ov}
The solution to OV indexing that we have described in Section~\ref{sec:second_solution} is limited by the tradeoff between the bitmaps for the lists in the top levels of the query graphs (lists in vertices $v_i$ such that $\alpha_i$ has a small number of ones) and the bitmaps for long lists in the bottom levels of the query graph. Therefore, we may improve the solution by making the lists in the bottom levels short, as for short lists we only save the elements themselves. We also note that we can also benefit from making the lists in the bottom levels very long, since their number is small. Consequently, the costly lists are those that are not too short and not too long.

In the solution we have presented in Section~\ref{sec:second_solution}, we pick a set $W$ of the indices for the query graph. Our solution works for any choice of $W$, but the question is whether there is a choice of $W$ that will make the list shorter or longer, so we can utilize it for a more compact solution to OV indexing. In the following lemma we show that for random input vectors that are uniformly distributed the probability for choosing $W$ such that there are lists that are not short is small.

\begin{lemma}
The size of every list $L_i$ in the query graph of the second solution of OV indexing is at most logarithmic w.h.p. for any choice of $W$ (the set of $k\log{n}$ indices) on random input vectors, where $k \geq 1$.
\end{lemma}
\begin{proof}
Let $h_W:S \rightarrow \{0,1\}^{k\log{n}}$ be a function that maps a $c\log{n}$-length boolean input vector $v \in S$ to a $k\log{n}$-length vector that contains only the elements of $v$ in the indices specified by $W$. The order of the elements is preserved. There are $\binom{c\log{n}}{k\log{n}}$ possibilities for choosing $W$. We consider all sets $L$ that are subset of $S$ and all their elements are mapped by $h_W$ to the same value (recall that each list $L_i$ in the query graph is also created by mapping the elements from $S$ using $h_W$ to some specific value). We want to analyse the probability that there exists $W$ such that there is a list $L$ of length at least $\ell$. This probability is the same as the probability that there exists $W$ such that there is a list $L$ of length \emph{exactly} $\ell$. This is trivial, since if there is a set $L$ whose length is larger than $S$ and all its elements are mapped to the same value by $h_W$ then we can take any subset of it of size exactly $\ell$ and all the elements in it are guaranteed to be mapped to the same value. The number of subsets of $S$ of length $\ell$ is $\binom{n}{\ell}$. Therefore, we have $\Pr[\exists W \subseteq [k\log{n}] \quad \exists L \subseteq S \quad |L| = \ell \quad \forall x,y \in L \quad h_W(x)=h_W(y)] \leq \binom{c\log{n}}{k\log{n}}\binom{n}{\ell}\Pr[\forall x,y \in L \quad |L| = \ell \quad h_W(x)=h_W(y)]$. To conclude our analysis we need to calculate the probability that all vectors in $L$ have the same elements in all the position specified by $W$. Let $S=\{v_1,v_2,...,v_n\}$ be the set of input vectors and let $v_i=(v_i^1,v_i^2,..., v_i^{c\log{n}})$ for $1 \leq i \leq n$. Moreover, we denote the elements of the set $W$ by $i_1,i_2,...,i_{k\log{n}}$. Some $\ell$ vectors $v_{j_1},v_{j_2},...,v_{j_{\ell}}$ will be mapped by $W$ to the same list $L$ if $v_{j_1}^{i_m} = v_{j_2}^{i_m} = ... = v_{j_\ell}^{i_m}$ for all $i_m \in W$. Therefore, for each $i_m \in W$ we have $\ell-1$ independent equations. As there are $k\log{n}$ elements in $W$ the total number of independent equations is $(\ell-1)k\log{n}$. The probability that $v_{j_a}^{i_m} = v_{j_b}^{i_m}$ for some $a,b$ and $m$ is exactly $1/2$ if the input vectors are random. Consequently, $\Pr[\forall x,y \in L \quad |L| = \ell \quad h_W(x)=h_W(y)] = 2^{-(\ell-1)k\log{n}}$. From this we get that $\Pr[\exists W \subseteq [k\log{n}] \quad \exists L \subseteq S \quad |L| = \ell \quad \forall x,y \in L \quad h_W(x)=h_W(y)] \leq \binom{c\log{n}}{k\log{n}}\binom{n}{\ell}2^{-(\ell-1)k\log{n}} \leq 2^{c\log{n}}2^{\ell\log{\frac{en}{\ell}}}2^{-(\ell-1)k\log{n}} = 2^{c\log{n}+\ell\log{e}+\ell\log{n}-\ell\log{\ell}-\ell k\log{n}+k\log{n}}$. It is easy to verify that if $\ell = 4c\log{n}$ the last expression is smaller than $1/n^4$.
\end{proof}

The last lemma guarantees that on random input vectors for any choice of $W$ the length of all lists in the query graph are supposed to be of length at most $4c\log{n}$ w.h.p. Therefore, instead of saving bitmaps for both lists in the top levels of the query graph and long lists in the bottom levels of the query graph, we need to save bitmaps just for the former as the latter do not exist. Consequently, for $c\log{n}$-length input vectors we just create the query graph with nodes representing $\log{n}$-length vectors and save bitmaps for the top $\delta \log{n}$ levels, for some $\delta>0$. The space required by these bitmaps is $n^{c-1+\epsilon}$, for some $\epsilon > 0$ that can be as small as possible by choosing appropriate small value for $\delta$. To conclude, we have obtained the following result:

\begin{theorem}
OV indexing on random input vectors can be solved in expected truly sublinear query time using $O(n^{c-1+\epsilon})$ space, for any $\epsilon > 0$.
\end{theorem}

This improved space complexity for random input makes it tempting to think that the same property holds even for worst case input. More specifically, it is enough to have just one $W$ that will map all input vectors to short or long lists. It turns out that for worst case input this cannot be achieved. In the following lemma we show how to create a worst case input vectors such that many lists in the query graph are neither too short nor too long.

\begin{lemma}
There exist $n$ $c\log{n}$-length input vectors such that for all $W \subseteq [k\log{n}]$ there are $\Theta(n)$ vectors that are mapped by $h_W$ to lists of size between $n^{1/6}$ and $n^{2/3}$
\end{lemma}
\begin{proof}
We start by showing that the lemma holds for $c=1.5$ and then generalize the construction to any $c>1$. We want to construct $n$ different vectors $v_1,v_2,...,v_n$. We split the $n$ vectors to groups of $\sqrt{n}$ vectors. The group $G_i$ contains the vectors $v_{i\sqrt{n}+1},v_{i\sqrt{n}+2},...,v_{(i+1)\sqrt{n}}$ for $0 \leq i \leq \sqrt{n}-1$. We partition each vector into 3 blocks of $0.5\log{n}$ bits. The block $B_j$ contains the bits in positions $1+0.5 j \log{n}, 2+0.5 j \log{n},..., 0.5 (j+1) \log{n}$. For every $G_i$ where $i \equiv j \pmod 3$ we put at the elements of block $B_j$ the values of all possible boolean vectors of length $0.5\log{n}$ such that each vector in $G_i$ receives a different value. In other words, for each group $G_i$ there is a block of $0.5\log{n}$ bits that separates between the vectors of the group. We call this block the separation block. The placement of the separation block is determined by the value of $i \bmod 3$. For each block $B_j$ we have explained how we fill the elements in the vectors of all groups $G_i$ such that $i \equiv j \pmod 3$. So we have another $2/3\sqrt{n}$ groups that are not filled yet. For each of these groups we pick a different $0.5\log{n}$-length boolean vector $v$ and fill the elements of block $B_j$ in all vectors of the group by the elements of $v$. This completes the construction of the input vectors.

Now, consider any choice of $W$ of size $\log{n}$ (the same reasoning works for any $1 \leq k < c$). We are guaranteed that there is a block $B_{j_1}$ that contains at most $1/3\log{n}$ elements from $W$ and a block $B_{j_2}$ that contains at least $1/3\log{n}$ elements from $W$. Therefore, if we focus on the vectors of all groups $G_i$ such that $i \equiv j_1 \pmod 3$. The bits from $W$ that are not in $B_{j_1}$ split these vectors into groups of size at least $\sqrt{n}$ (as all vectors in the same group are not split by bits not in $B_{j_1}$) and at most $n^{2/3}$ (as each of the $1/3\log{n}$ bits from $B_{j_2}$ doubles the number of groups). Each of the bits in $B_{j_1}$ divides the groups into two halves. Therefore, the number of vectors in each group is at least $n^{1/6}$ and at most $n^{2/3}$. The total number of vectors in these groups is $O(n)$.

For input vectors of length $c\log{n}$ if $c = 0.5a$ for some integer $a$, then we can have $a$ blocks of size $0.5\log{n}$ and the placement of the separation block will be determined by $i \bmod a$. The rest of the construction and the analysis is similar to what we do in case $a=3$.

If $c=0.5a+b$ for some integer $a$ and $0 < b < 0.5$, then all blocks are of size $0.5\log{n}$ except the last block that will be of size $b\log{n}$. All the groups that their separation block is of size $0.5\log{n}$ will be of size $\sqrt{n}$. All groups that their separation block is of size $b\log{n}$ will be of size $n^{b}$ (the last groups can be smaller if there are not enough vectors left to fill them). Except this change all the other parts of the construction and anlysis are similar to the case  of $1.5\log{n}$-length vectors.
\end{proof}

In the last lemma we can obtain values other than $n^{1/6}$ and $n^{2/3}$ by changing the basic block size from $0.5\log{n}$ to some other $r\log{n}$ for $r>0$.

This demonstrates that for worst-case input vectors, as opposed to random input vectors, there can be $O(n)$ vectors that are mapped to lists that are neither too short nor too long. The exact size can be controlled by proper choices of block and group size.

\section{Further Research}
In this paper we presented several algorithms to solve OV indexing that obtain truly sublinear query time and require $O(n^{c-\gamma})$ space for some constant $0<\gamma<1$. For random input vectors we demonstrated in Section~\ref{sec:random_ov} how to obtain sublinear query time solution to OV indexing using $O(n^{c-\gamma})$ for \emph{any} $0<\gamma<1$. We note that the preprocessing time of all algorithms is polynomial in $n$.

The main question regarding OV indexing, following this paper, is can one obtain a sublinear query time solution to OV indexing that requires only $O(n^{c-1})$ space. This question is interesting even if we allow an unlimited preprocessing time. We conjecture that there is no such solution to OV indexing:

\begin{conjecture}
There is no truly sublinear query time solution to OV indexing that requires only $O(n^{c-1})$ space.
\end{conjecture}

Even if that conjecture is false, it is of utmost interest to find the exact lower bound on the space requirements of OV indexing for both unlimited and polynomial preprocessing time. Finding the exact space requirements can be used to obtain conditional lower bounds on the \emph{space} complexity of many problems known to be conditionally hard in terms of time based on OV.       

\bibliographystyle{plain} \bibliography{ms}

\newpage
\appendix
\section*{Appendix}

\section{Conditional Time Hardness for OV indexing}\label{sec:hardness_of_OV_indexing}
The following conditional lower bound on OV was proved by Williams~\cite{Williams05} (see also~\cite{WY14}):

There is no $\delta>0$ such that for all $c \geq 1$ there is a solution to Orthogonal Vectors problem on $n$ vectors in $\{0,1\}^{c\log{n}}$ running in $O(n^{2-\delta})$ time, unless SETH is false.

Given this conditional lower bound it is easy to get the following conditional lower bound on the OV indexing problem, by just solving OV using a solution to OV Indexing:

There is no $\delta>0$ such that for all $c \geq 1$ there is a solution to Orthogonal Vectors Indexing problem on $n$ vectors in $\{0,1\}^{c\log{n}}$ having $O(n^{2-\delta})$ preprocessing time and $O(n^{1-\delta})$ query time, unless SETH is false.

However, we can get better conditional lower bound by reducing directly from $k$SAT using the same ideas as in~\cite{Williams05} and~\cite{AWW14}:

\begin{lemma}
There is no $\alpha$, $\delta>0$ such that for all $c \geq 1$ there is a solution to Orthogonal Vectors Indexing problem on $n$ vectors in $\{0,1\}^{c\log{n}}$ having $O(n^{\alpha})$ preprocessing time and $O(n^{1-\delta})$ query time, unless SETH is false.
\end{lemma}

\begin{proof}
Given a $k$-CNF formula $\phi$ on $n$ variables, we use the sparsification lemma~\cite{IP01} to obtain $m=O(2^{\epsilon n})$ $k$-CNF formulas $\phi_1,\phi_2,...,\phi_m$ such that one of them is satisfiable iff the $\phi$ is satisfiable. The number of clauses in each $\phi_i$ is at most $f(k,\epsilon) n$ for some function $f$.
Let $x_1,x_2,...,x_n$ be the variables in $\phi$. We focus on some specific formula $\phi_i$ having the clauses $y_1,y_2,...,y_{\ell}$. We generate all possible partial assignments to the first $n/t$ variables. For each partial assignment, we generate a boolean vector $v_k$ such that the element at position $j$ in $v$ is 0 if the partial assignment satisfies $y_j$ and 1 otherwise. These vectors become the input vectors of an instance of OV indexing. As a query we generate a boolean vector $u$ for each partial assignment to the last $n(1-1/t)$ variables. It is easy to verify that $\phi_i$ is satisfied iff there is a query vector $u$ that is orthogonal to some $v_k$. Therefore, we can decide if $\phi$ is satisfiable by $O(2^{\epsilon n})$ instances of OV indexing. The length of the input vectors in each instance is at most $f(k,\epsilon)n$, the total number of input vectors is $2^{n/t}$ and the total number of queries is $2^{n(1-1/t)}$. Let $O(N^{\alpha})$ be the preprocessing time for OV indexing on $N$ input vectors and let the query time be $O(N^{1-\delta})$ for some $\delta>0$. In our case $N=2^{n/t}$. The total running time of our reduction is $2^{\epsilon n}(2^{n\alpha/t}+2^{n(1-1/t)}2^{n/t(1-\delta)})$. Setting $\epsilon<\delta/2t$ and $t>\alpha/(1-\delta)$ (choosing $t$ to be the smallest number that satisfies this inequality and guarantees than $n/t$ is an integer) we get a total running time of $O^{*}(2^{(1-\epsilon ')n})$ for solving $k$SAT for some $\epsilon'>0$.
\end{proof}

We note that a similar result can also be obtained by reducing from Orthogonal Vectors. This can be achieved by splitting the set $S$ of input vectors to many sets with small amount of vectors and then querying each one of them by all vectors of $S$. 

\newpage

\section{Approximation of Binomial Coefficients} \label{sec:approx_binomial}

Throughout this paper we have to calculate the value of binomial coefficients of this form $\binom{m\log{n}}{k\log{n}}$. The next lemma gives an approximation for this binomial coefficient based on Stirling's approximation.

\begin{lemma}
$\binom{m\log{n}}{k\log{n}} \approx \tilde{O}(n^{m\log{m}-k\log{k}-(m-k)\log{(m-k)}})$ for all $m>k>0$.
\end{lemma}
\begin{equation*}
\begin{split}
& \binom{m\log{n}}{k\log{n}} = \frac{(m\log{n})!}{(k\log{n})!((m-k)\log{n})!} \\
& \approx \frac{\sqrt{2\pi m\log{n}}(\frac{m\log{n}}{e})^{m\log{n}}}{\sqrt{2\pi k\log{n}}(\frac{k\log{n}}{e})^{k\log{n}}\sqrt{2\pi (m-k)\log{n}}(\frac{(m-k)\log{n}}{e})^{(m-k)\log{n}}} \\
& = \sqrt{\frac{m}{2\pi k(m-k)\log{n}}}\frac{m^{m\log{n}}}{k^{k\log{n}}(m-k)^{(m-k)\log{n}}} \\
& = \sqrt{\frac{m}{2\pi k(m-k)\log{n}}}\frac{n^{m\log{m}}}{n^{k\log{k}}n^{(m-k)\log{(m-k)}}} \\
& = \sqrt{\frac{m}{2\pi k(m-k)\log{n}}}n^{m\log{m}-k\log{k}-(m-k)\log{(m-k)}} \\
& =\tilde{O}(n^{m\log{m}-k\log{k}-(m-k)\log{(m-k)}})
\end{split}
\end{equation*}

\begin{corollary}
$\binom{m\log{n}}{\log{n}} \approx \tilde{O}(n^{m\log{m}-(m-1)\log{(m-1)}})$ for all $m>0$.
\end{corollary}

\end{document}